\documentclass[twoside]{article}

\usepackage{arxiv}

\usepackage[linesnumbered,ruled,noend]{algorithm2e}
\usepackage{amsfonts}
\usepackage{amsmath}
\usepackage{amsthm}
\usepackage[justification=centering]{caption}
\usepackage{enumitem}
\usepackage{gensymb}
\usepackage{multirow}
\usepackage[round]{natbib}
\usepackage{subfig}
\usepackage{tabularx}
\usepackage{tikz}
\usepackage{todonotes}
\usetikzlibrary{calc,intersections}

\tikzset{
    right angle quadrant/.code={
        \pgfmathsetmacro\quadranta{{1,1,-1,-1}[#1-1]}
        \pgfmathsetmacro\quadrantb{{1,-1,-1,1}[#1-1]}},
    right angle quadrant=1,
    right angle length/.code={\def\rightanglelength{#1}},
    right angle length=2ex,
    right angle symbol/.style n args={3}{
        insert path={
            let \p0 = ($(#1)!(#3)!(#2)$) in
                let \p1 = ($(\p0)!\quadranta*\rightanglelength!(#3)$),
                \p2 = ($(\p0)!\quadrantb*\rightanglelength!(#2)$) in
                let \p3 = ($(\p1)+(\p2)-(\p0)$) in
            (\p1) -- (\p3) -- (\p2)
        }
    }
}

\newtheorem{theorem}{Theorem}
\newtheorem{lemma}{Lemma}
\newtheorem{corollary}{Corollary}

\DeclareMathOperator*{\argmax}{argmax}

\begin{document}

%

%

\twocolumn[

\aistatstitle{Revealing the Basis:
Ordinal Embedding through Geometry}

\aistatsauthor{ Jesse Anderton \And Virgil Pavlu \And Javed Aslam }

\aistatsaddress{
	Northeastern University \\
	Boston, Massachusetts USA \\
	jesse@ccs.neu.edu	
\And
	Northeastern University \\
	Boston, Massachusetts USA \\
	vip@ccs.neu.edu	
\And
	Northeastern University \\
	Boston, Massachusetts USA \\
	jaa@ccs.neu.edu	
} ]

\begin{abstract}

Ordinal Embedding places $n$ objects into $\mathbb{R}^d$ based on
comparisons such as ``$a$ is closer to $b$ than $c$.''
Current optimization-based approaches suffer from scalability problems and an
abundance of low quality local optima.
We instead consider a computational geometric approach based on selecting
comparisons to discover points close to nearly-orthogonal ``axes'' and embed
the
whole set by their projections along each axis.
We thus also estimate the dimensionality of the data.
Our embeddings are of lower quality than the global optima of
optimization-based approaches, but are more scalable computationally and more
reliable
than local optima often found via optimization.
Our method uses $\Theta(n d \log n)$ comparisons and $\Theta(n^2 d^2)$
total operations, and can also be viewed as selecting constraints for an
optimizer which, if successful, will produce an almost-perfect embedding for
sufficiently dense datasets.

\end{abstract}

\section{Introduction}
\label{sec-intro}

Ordinal Geometry is a family of related problems that deals with a collection
of objects which are presumed to lie in some metric space with an unknown
metric. 
Unable to directly use features/positions, distances, or similarity scores, we
instead attempt to either recover the metric (positions or distances), or to
perform tasks such as
clustering or classification, using only ordinal comparisons of the form,
``Object $a$ is closer to object $b$ than to object $c$.''

For example, a collection of movies can be analyzed without extracting
features by asking movie-goers questions like,
``Is \texttt{Star Wars} more similar to \texttt{The Matrix} or to \texttt{Star
Trek}?''.
Similarly, one might analyze songs or musicians without having the
musical expertise to design a feature representation for either
by soliciting user preferences like,
``If you love \texttt{The Rolling Stones},
you will prefer \texttt{Bruce Springsteen} to \texttt{Madonna}.''
Similarity on food dishes might depend on factors like smell and taste, yet
it's not obvious how to obtain consistent features per dish; still, people
have no trouble comparing foods.


These comparisons are commonly gathered from human assessors, who often find it
easier to report on relative similarity than to provide explicit distances.
In other cases, ordinal information can be inferred from known features as a
regularization, dimensionality reduction, or representation learning technique.
Ordinal Embedding, which seeks to recover the metric, has much in common with
Metric/Kernel Learning but differs in that we generally have no underlying
features.

In this paper we address the problem of recovering the object pairwise
distances (or simply ``recovering the metric'') using a minimum number of
comparisons, assuming the underlying metric is Euclidean (i.e. the objects
exist in an unknown Euclidean space) and that comparisons are always answered
correctly.
Note that recovering the full matrix of pairwise distances is mathematically
equivalent to an embedding into a space isomorphic with the original data space
up to scaling, rotation, reflection, and translation.
With the metric in hand we can run standard machine learning tools,
either by employing positions as vectors of latent features or by using the
metric as a similarity kernel.

\textbf{Problem Context.}
\citet{Jamieson:2011bt} proved a lower bound of $\Omega(n d \log n)$
adaptively selected comparisons to fully recover the underlying metric space.
While our work only fully recovers the metric for certain ``ideal''
datasets, it is the first published non-trivial algorithm with an
upper bound matching the lower bound and thus establishes a performance
baseline for future algorithms.

The theoretical results of \citet{Kleindessner:2014wu} and
\citet{AriasCastro:2015wk} prove that full metric recovery is possible given a
total ordering of the $k$-nearest neighbors ($k$NN) of each point, when the
dataset meets certain distributional requirements and for
$k = \omega((n \log n)^{1/2})$, but do not propose an algorithm to do so.
\citet{Hashimoto:2015wj} provide an algorithm to fully recover the metric when
$k = \omega(n^{2/(d+2)}(\log n)^{d/(d+2)})$,
and which performs well even when $k = \log n$.
The problem of comparison selection for full metric recovery is thus reduced to
efficiently identifying the $k$NN of each point.
However, we are not aware of any $k$NN algorithms using only $O(n d \log n)$
comparisons in the ordinal geometry setting.
Despite the vast literature on identifying the $k$NN for a set,
all the algorithms we have surveyed rely on known positions or distances.

The trivial way to obtain all possible comparisons is to sort all $n$
points by increasing distance to each of the $n$ possible ``heads,'' for
$O(n^2 \log n)$ comparisons.
Alternatively, a selection algorithm can be used to obtain all objects' $k$NN
using $O(n^2)$ comparisons.
In either case, an embedding algorithm can then be used to (hopefully) recover
the metric.
When $d = \Omega(n)$, this is the best that can be accomplished.

The present work contains the first non-trivial upper bound on the problem
when $d = o(n)$, matching the lower bound to
(1) produce an embedding to recover the metric down to scaling for
``ideal'' datasets,
and to
(2) produce the input required for a high-quality embedding for many
realistic ``non-ideal'' datasets.
We also study the properties a dataset needs for our method to work well,
showing where future work can generalize our approach to further classes
of problems.

Most existing work on ordinal embedding, such as Soft Ordinal Embedding, or SOE
\citep{Terada:2014tr}, focuses on optimizing a non-convex Machine Learning
objective within a Euclidean space of some assumed dimensionality.
We have empirically found that even when all comparisons are consistent with a
Euclidean metric and the optimizer is given the correct dimensionality,
optimization-based methods for ordinal embedding often
struggle to find a global optimum or to deal with large datasets.
We have also found that local optima tend to produce rankings that are
nearly-random (when sorting by distance to any given object).

\begin{figure*}[ht]
\centering
\includegraphics[width=.3\textwidth]{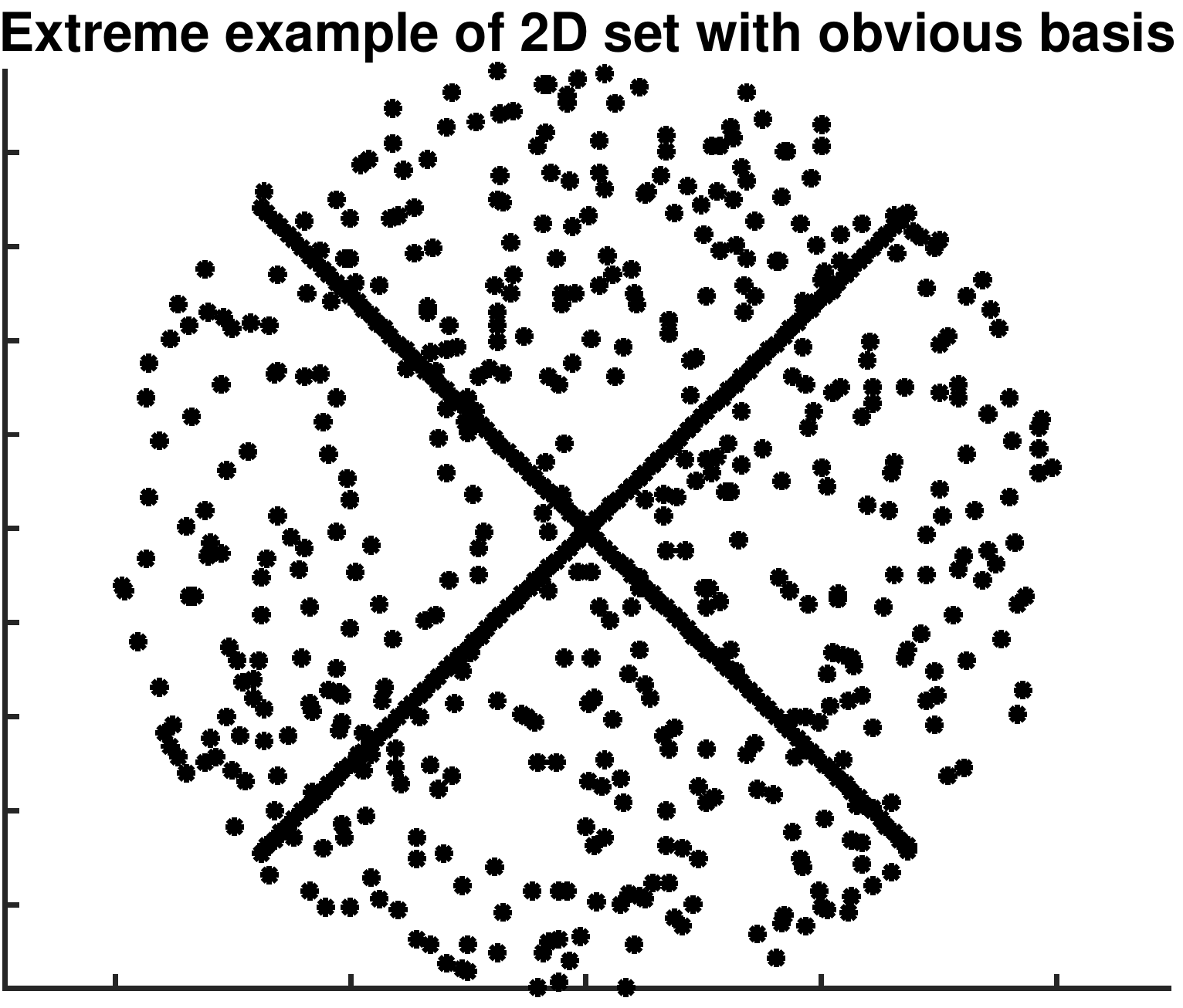}
\quad
\includegraphics[width=.3\textwidth]{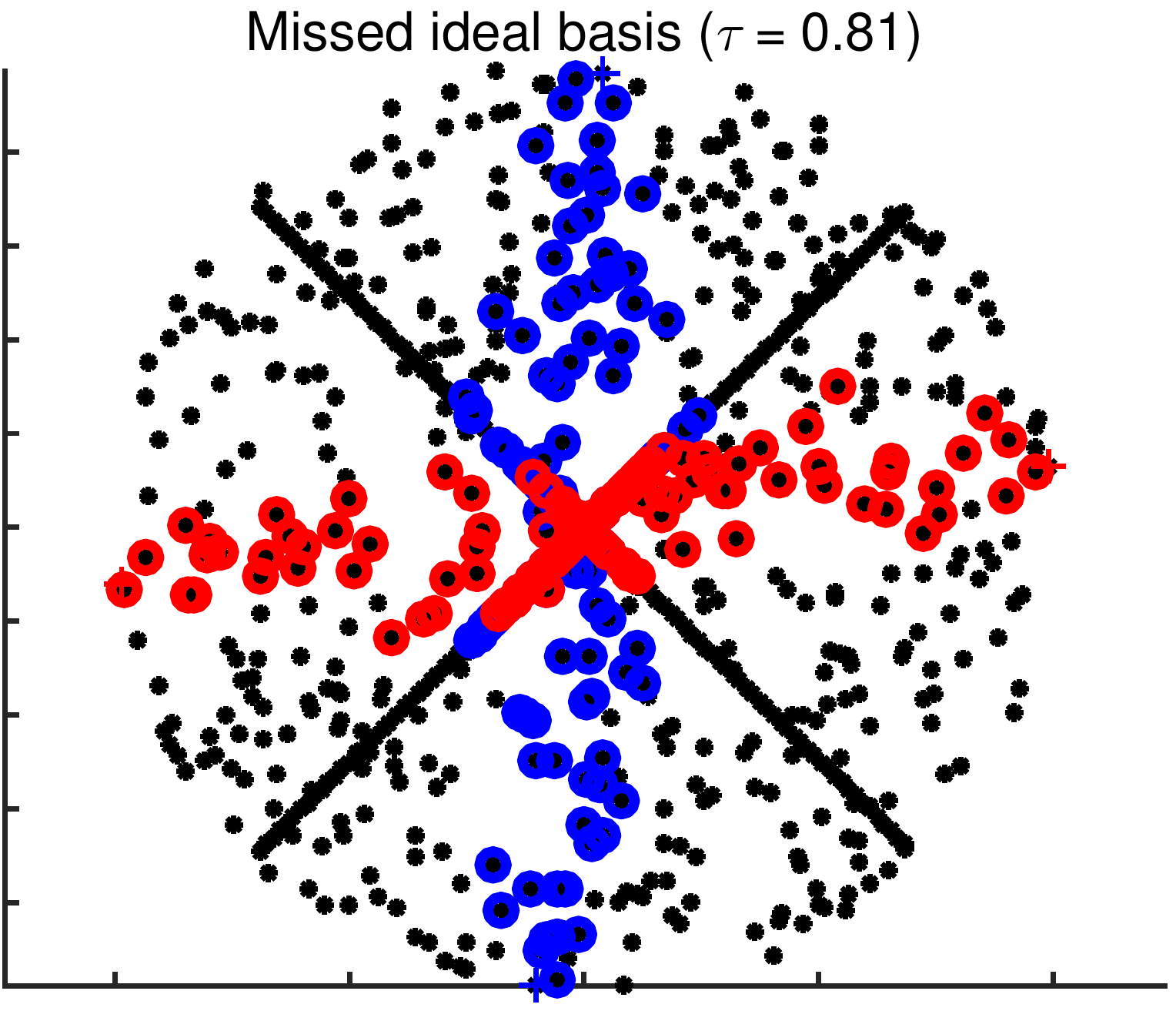}
\quad
\includegraphics[width=.3\textwidth]{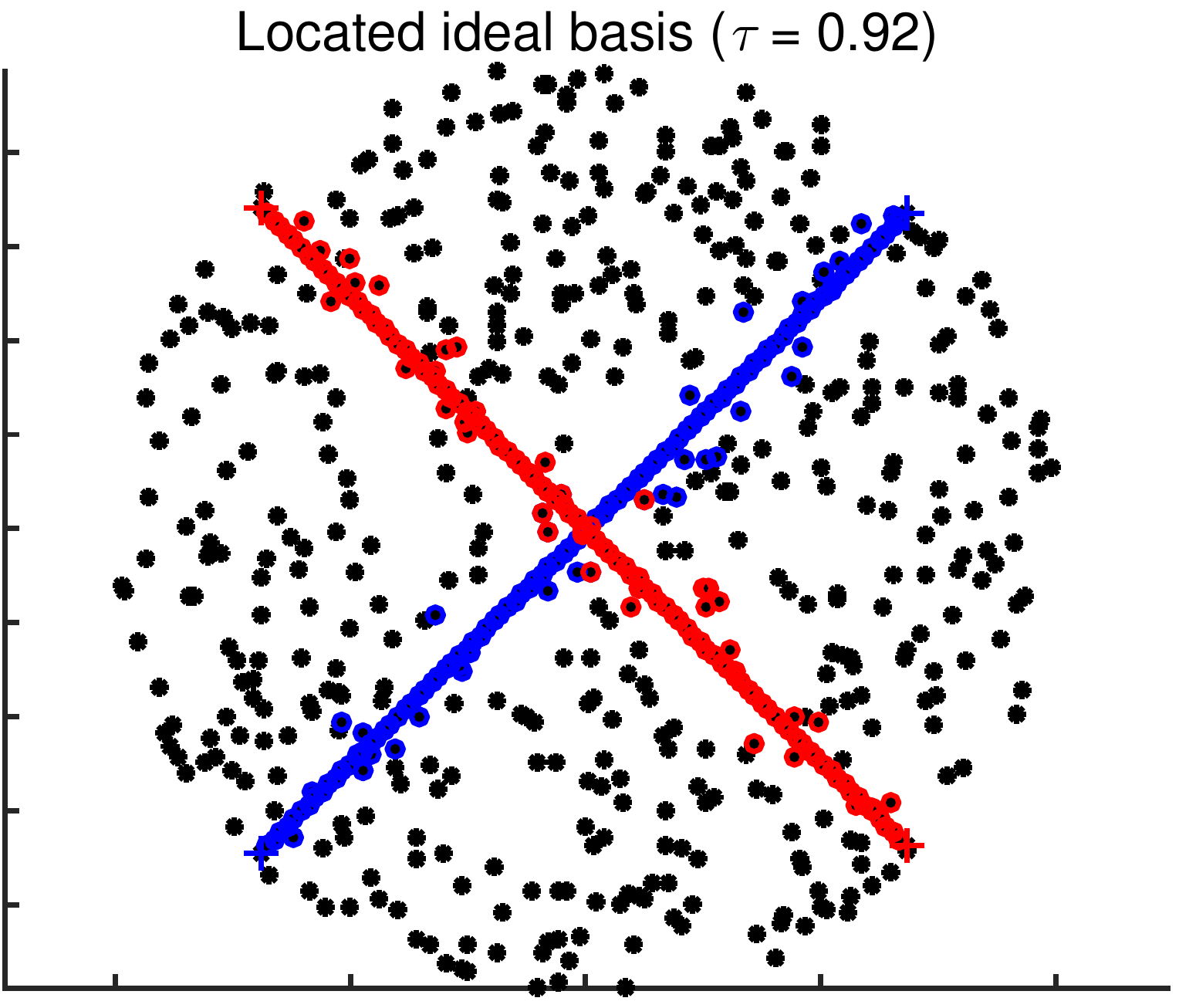}
\caption{2D points set $X$ (left) includes two subsets of colinear, dense, evenly-spaced, points that make obvious long-enough axes. Middle: two axes red and blue found by 
our algorithm are not the best basis, but reasonable; $\tau$ is the mean Kendall's $\tau$ between true rankings and basis-estimated rankings. Right: ideal axis red, blue also include few other points (criteria in Fig 2) since the gap on the colinear points is not small enough; thus $\tau<1$.
}
\label{fig-missed-perfect}
\vspace{-1.5em}
\end{figure*}

The present work takes on the challenges of state-of-the-art methods, namely
scalability and reliability, by
considering the underlying geometry of the problem, and avoiding optimization
at some expense of performance.
Our main contributions are:
\begin{itemize}[nosep]
	\item an algorithm to quickly produce embeddings of moderate to high quality
		using computational geometry rather than optimization;
	\item the first published algorithm to adaptively select
		a number of comparisons within a constant factor of the lower bound which
		achieves nearly-perfect embeddings for realistic datasets
		($d \ll n$ and ``not too sparse''); and
	\item the first published latent dimensionality estimator for ordinal
		datasets.
\end{itemize}
\vspace {-1ex}
We also introduce subroutines which provide convex hull estimation,
perpendicular line discovery, and a test of affine independence.

\textbf{Our approach.}
We build our algorithm on the following intuition.
Suppose that for some finite dataset $X \subset \mathbb{R}^d$ we also had
access to $d$ orthogonal ``axes'', not as actual lines or segments, but in the
form of $d$ subsets of points $A_1,\dots,A_d$. The points on each axis are
colinear, dense (max gap $\epsilon$), and evenly-spaced. Additionally each
$A_i$ is ``long enough'' so that every point in $X$ has a ``projection'' on it.

Each axis $A_i$ contains two special
``endpoints''; the $A_i$ points have the \emph{order-consistent property} that
sorting by increasing distance from one endpoint
is equivalent to sorting by decreasing distance from the other endpoint.
For each point $x$ in $X$, we can approximate its geometric projection to an
axis $A_i$ by locating the closest point on $A_i$ to $x$ with a binary search
(using triplet comparisons), for a total cost of
$\Theta(n d \log m)$ comparisons, where $m = \max_i |A_i|$.
The coordinate is given by the rank of that closest axis point from an
arbitrary endpoint.
Note that if the axis points are evenly distributed, then
each ordinal coordinate point is within $\epsilon$ of the (rescaled) true
projection on that axis; if $\epsilon$ is small enough,
this embedding recovers the original metric space with arbitrary precision.

Throughout the paper, we refer to such a set of points 
($A_i$) as an ``axis,'' and to the collection of axes as a ``basis.''
In general, a real dataset will not contain evenly-spaced points along the axes
of a perfect basis, but we achieve good performance with
axis points which are only approximately collinear and axes which are
only approximately orthogonal.
We will show
how to find such axes as subsets of $X$.
The axes we find, as sets, have the order-consistent property between
endpoints, so ranking is non-ambiguous.
Other desirable properties for high quality ordinal coordinates are:
\begin{enumerate}[label={(\arabic*)},align=left,nosep]
\item enough axes to account for differences between points,
\item axes that are closer to being orthogonal,
\item axes that extend to the boundaries of a bounding box for the set,
\item points along all axes that are evenly-spaced, and
\item gaps between points on each axis that are smaller than gaps between
	off-axis points.
\end{enumerate}
Note that we do not require the axes to intersect (i.e. for $d>2$ we do not
need to have an origin point).

We address
(1) through estimating affine independence,
(2) through the symmetry of sphere intersections,
and (3), roughly, through finding points close to the boundary of the convex
hull of $X$ (denoted $conv(X)$),
but do not directly address (4) or (5) in the present work.
For this reason, our method works best under relatively smooth density
conditions and does not perform as well with high-dimensional datasets or
datasets with rapidly-varying densities.
See Figure~\ref{fig-missed-perfect} for the impact of missing the ``ideal''
basis.

If $\epsilon$ is the radius
of the largest open ball in $conv(X)$ which contains no members of $X$, we say
$X$ is $\epsilon$-dense.
It is easy to show that as $\epsilon \rightarrow 0$
the basis our algorithm finds converges to a perfect basis and thus we fully
recover the underlying metric.
When our basis is imperfect, we can still use it to efficiently
perform tasks such as finding the $k$-nearest neighbors of each member of $X$,
or producing input to the user's favorite embedding routine (such as SOE).

\section{Finding an Approximate Basis}
\label{sec-basis}

Algorithm~\ref{fig-alg-basis} is our basis-finding algorithm.
Each of the $\hat{d}$ axes is formed by choosing pairs of endpoints near the
boundary of $X$ and finding points from $X$ close to their (linear) convex
hull.
The first axis is formed using two points opposite each other on the boundary
of $conv(X)$, using that the most-distant point from any member of $X$ is on
this boundary.
After that, we find new axes which are likely to be orthogonal to the previous
set of axes by identifying points which are far from the convex hull of the
previous axis endpoints (see Section~\ref{sec-basis-axes}).
In our algorithm, $r_{p_i}[\cdot]$ returned by
\texttt{SortForHead}($p_i,oracle$)
is the array of the ranks of all objects $x \in X$ sorted by distance from
endpoint $p_i$.

\begin{algorithm}[h]
\caption{ChooseBasis($n,oracle$)}
\label{fig-alg-basis}
\SetKw{Break}{break}
\SetKwInOut{Input}{Input}
\SetKwInOut{Output}{Output}
\SetKwFunction{SortForHead}{SortForHead}

\Input{$n$ is the number of objects in the collection
\\ \hangindent=2cm
$oracle(a,b,c)$ decides whether $b$ or $c$ is closer to $a$.}
\Output{A basis $(A_1,\dots,A_{\hat{d}})$, where each axis $A_i$
	lists points near a line crossing the dataset.}
$z \gets $ \emph{a randomly selected point} \;
$p \gets $ \emph{furthest point from $z$} ; \tcp{first axis endpoint}
$r_{p}[\cdot] \gets \SortForHead(p, oracle)$ \;
$P \gets \{ p \}$ \;
\For{$i \gets 1,3,5,\dots,n$}{
	\tcp{Complete the $(i+1)/2$ axis}
	$p_i \gets p$ \;
	$L \gets \{ x : r_{p_j}[x] \le r_{p_j}[p_i], \forall j<i \}$ ; 
		\hspace{3ex} \tcp{Lens}
	$p_{i+1} \gets \argmax_{x \in L} r_{p_i}[x]$ ; \tcp{Apex oppos. $p_i$} 
	$r_{p_{i+1}}[\cdot] \gets \SortForHead(p_{i+1}, oracle)$ \;
	$\hat{d} \gets (i+1)/2$ \;
	$A_{\hat{d}} \gets \widehat{conv}(\{p_i, p_{i+1}\})$ \;
	\tcp{Verify candidate for next axis}
	$p \gets $ \emph{point ``above'' max \# of points in $\widehat{conv}(P)$}\;
	\If{\emph{no point is ``above'' any other}}{
		\Break \;
	}
	$r_{p}[\cdot] \gets \SortForHead(p, oracle)$ \;
	$P \gets P \cup \{p\}$ \;
	\If{$\hat{d} = 1$}{
		$P \gets P \cup \{p_{i+1}\}$ \;
	}
	\tcp{Test affine independence}
	\If{$\widehat{conv}(P) = \cup_{z \in P} \widehat{conv}(P \setminus \{z\})$}{
		\Break \;
	}
}
\Return $(A_1,\dots,A_{\hat{d}})$ \;
\end{algorithm}

We fully describe our algorithm in the following subsections,
and analyze its cost here.
Line 2 uses $n-2$ comparisons.
Each call to \texttt{SortForHead} sorts all objects in $\Theta(n \log n)$ 
comparisons.
We make two calls per axis, so we use $\Theta(n \hat{d} \log n)$ comparisons
(within the theoretical bound as long as $\hat{d} = O(d)$).

The $\widehat{conv}$ function (Eq.~\ref{eq-conv-estimate})
estimates the convex hull of a set of points based on their rankings of
the other points.
It uses $O(n^2 \hat{d})$ operations (not comparisons), because it has to
iterate over all $2\hat{d}$ rankings for each point.
The calculation on Line 12 to find the next candidate takes $O(n\hat{d})$
operations to scan the axis endpoints' rankings of each point in $X$.
Thus, \texttt{ChooseBasis} uses $\Theta(n^2 \hat{d}^2)$ total operations.

\subsection{Choosing Axes}
\label{sec-basis-axes}

We identify axes by choosing axis endpoints which are as far as
possible from the convex hull of the previous axis endpoints. 
Let $P = \{p_1,\dots,p_{2\hat{d}}\}$ be the set of endpoints for the $\hat{d}$
axes found so far.
A straightforward approach, different than ours, is a
\emph{farthest-rank-first traversal of $X$} (FRFT) adapted from
\citet{Gonzalez:1985kk}:
choose endpoints as far as possible (in rank) from the previous axis endpoints,
i.e. by 
$\argmax_{x \in X} \{\min_{p \in P} r_p[x]\}$.

This forms a rank-based approximation of an $\epsilon$-net, which can be
used to approximate the geometric distribution of a set of points.
The axis endpoints thus found are well-separated from each other,
and tend to lie closer to regions of higher density than the points of an
$\epsilon$-net.
However, our testing occasionally found that some of the resulting axes
are nearly parallel.

We present a more reliable approach. We first discuss the notions of a
\emph{lens} and its \emph{apex}, and \emph{above}-ness.
Suppose we have selected a single axis with endpoints $p_1$ and $p_2$, and that
for all $x \in X$ we find that
$d(p_1,x) \le d(p_1,p_2)$ and $d(p_2,x) \le d(p_2,p_1)$.
Then the set $X$ lies within the \emph{lens} between $p_1$ and
$p_2$ --- the intersection of the closed balls centered on $p_1$ and $p_2$ with
radii both equal to $d(p_1,p_2)$.
Future ideal axis endpoints will lie close to the
\emph{apex} of this lens: the set of points $\mathcal{A}(\{p_1,p_2\})$, where
\begin{align} \label{eq-lens-apex}
\mathcal{A}(P) \equiv \{ x \in \mathbb{R}^d : \forall p \in P,
	d(x,p) = \max_{q \in P} d(p,q) \}
\end{align}
is the intersection of the (hollow) spheres centered on vertexes $P$ that
surround $P$. In general on $d$ dimensions, $d$ points will create an apex of 2
points, $d-1$ points an apex of a circle, $d-2$ points a sphere, etc.
A ``lens'' can similarly be formed by choosing a point $a \in X$ as
an apex and using the distances $d(p,a)$ for each $p \in P$ as the ball radii
(``lens'' is informal here, as it is the intersection of $|P|$
and not strictly 2 balls).

When points near the lens apex exist in $X$, these points are ideal choices
because they are as far as possible from, and form orthogonal lines to, the
line between $p_1$ and $p_2$.
FRFT will naturally choose these points, when possible.
However, these points only exist in $X$ when the range
along all dimensions is almost equal.

We say that point $p$ is \emph{above} point $q$ with respect to some set $P$
if for all $a \in P$ we have $d(a,p) > d(a,q)$.
See Figure~\ref{fig-above} for an example.
This means that $q$ lies in the ``lens'' of $P$ with $p$ as its apex.
By Theorem~\ref{prop-ball-covering} below, $p$ can thus not
be in $conv(P)$.
If the set $P$ is fixed, above-ness is transitive:
if $p$ is above $q$ and $q$ is above $z$, that implies $p$ is above $z$.

\begin{figure}[ht]
\centering
\includegraphics[width=.35\textwidth]{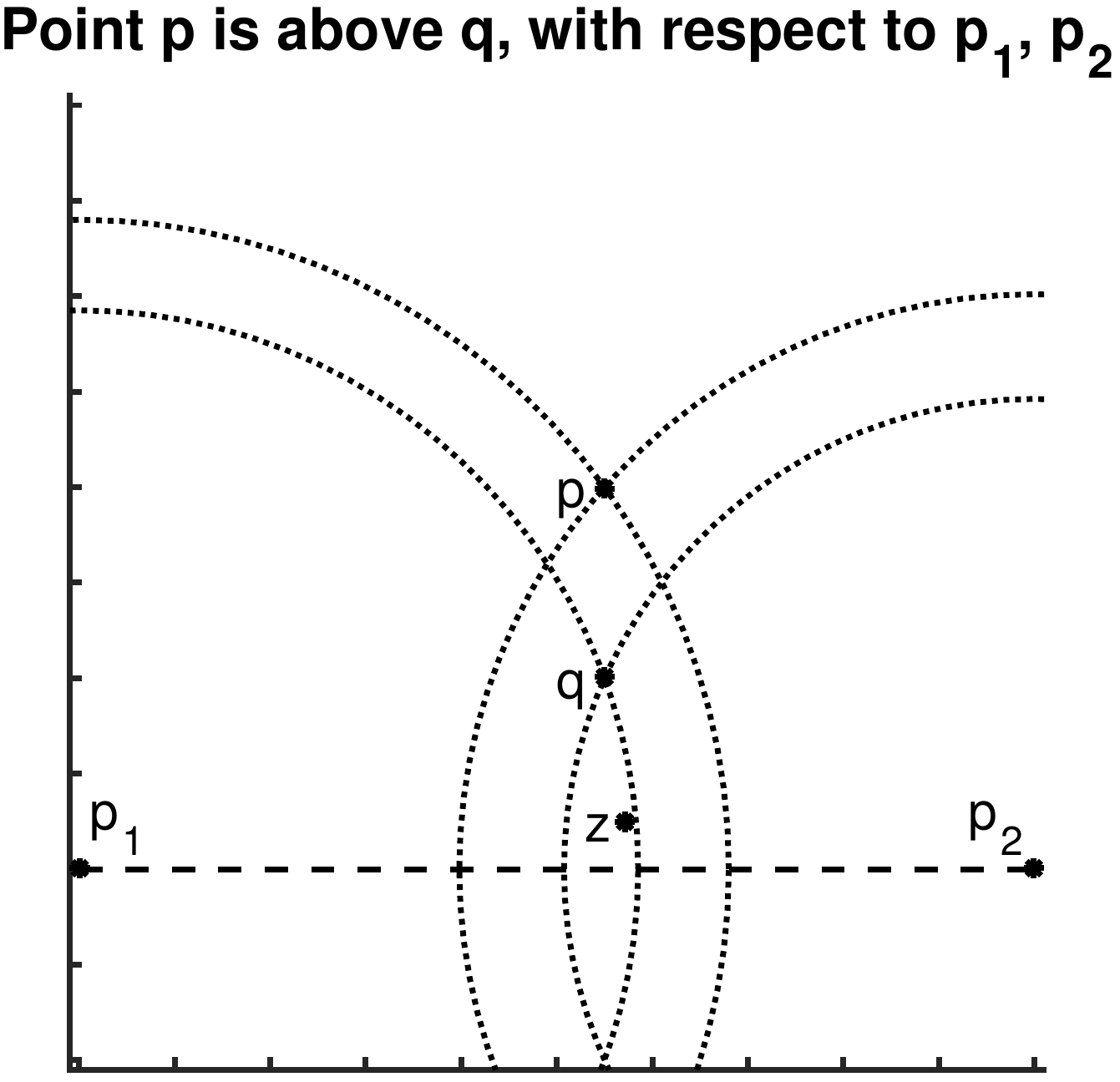}
\captionsetup{justification=raggedright}
\caption{Point $p$ is ``above'' $q$ because $q$ is found in the
intersection of balls centered on $p_1$ and $p_2$ and extending to $p$.
This implies that $p$ is farther from the line than $q$ and thus not in
$conv(\{p_1,p_2\})$.
}
\label{fig-above}
\vspace{-1em}
\end{figure}

\textbf{Axis endpoints selection}. The first new axis endpoint $p_{i}=p$
(Algorithm~\ref{fig-alg-basis} line 12) is chosen to be
above a maximal number of points in the convex hull of prior endpoints.
If there is a point close to the apex of the ball
intersection, it will be above nearly the entire set and will be chosen as the
maximum (and our algorithm will make the same choice as FRFT).
If there is no such point then we will choose a point which is above as many
points as possible, hoping for a point above a dense region of $X$ which
also lies close to the apex of the ball intersection.
The second point $p_{i+1}$ (line 8) is chosen to approximate the apex opposite
the first endpoint $p_{i}$ in the lens with $p_i$ as an apex.

\textbf{Dimensionality estimate $\hat{d}$}. We stop adding axes as soon as our
next axis endpoint $p$ does not appear
to be affinely-independent of the previous axis endpoints.
Our test for affine independence relies on Carath\'{e}odory's Theorem, which
states that any point in the convex hull of a set of points in $\mathbb{R}^d$
can be expressed as a convex combination of just $d+1$ or fewer of them.
We define a set $P$ containing both endpoints for the first axis, plus
the first endpoint of each additional axis and our new candidate $p$.
We have $|P| = \hat{d}+2$.
Suppose we have already found $d$ dimensions, and the new axis endpoint is
simply not in the convex hull of the previous endpoints.
Then any point in $conv(P)$ is also in the convex hull of some
set of all but one of the points in $P$.
On the other hand, if $d > \hat{d}$ then these extra hulls are the faces of a
higher-dimensional manifold, and we expect some of the points in that manifold
to be far enough from the convex hulls to not be included in our
$\widehat{conv}$ estimates.

We prove that our dimensionality estimate $\hat{d}$ is always at most the true
dimensionality $d$ and that it converges to $d$ as $\epsilon \rightarrow 0$ in
Theorem~\ref{prop-dim-converges} in Appendix~\ref{app-basis}.
See Table~\ref{tbl-dim-estimates} for dimensionality estimates on various
datasets.
Since the number of points in each dataset is the same, as the dimensionality
increases the density constant $\epsilon$ grows.
This causes $\widehat{conv}$ to be less precise and leads us to
underestimate the dimensionality.

\begin{table}
\caption{Dimensionality Estimates
\\
(1,000 points, avg. of 100 runs)
}
\label{tbl-dim-estimates}
\begin{center}
\begin{tabular}{lrrrrrrr}
{\bf True $d$:} & 1 & 2 & 3 & 5 & 8 & 10 & 20 \\
\hline
Ball & 1 & 2 & 2.11 & 3.66 & 4.22 & 4.54 & 5.53 \\
Cube & 1 & 2 & 2.37 & 3.74 & 4.44 & 4.58 & 4.78 \\
Gaussian & 1 & 2 & 2.98 & 3.91 & 4.44 & 4.54 & 4.52 \\
Sphere & 1 & 1 & 2 & 3.09 & 3.85 & 4.08 & 4.93 \\
\end{tabular}
\end{center}
\vspace{-1.5em}
\end{table}

\subsection{Convex Hull Estimation}
\label{sec-conv-hull}

Our convex hull estimation $\widehat{conv}$ works because any union of balls
which all coincide in some point must contain the convex hull of the ball
centers.
We prove this as Theorem~\ref{prop-ball-covering} in
Appendex~\ref{app-convex-rankings}.

Suppose we want to identify points from $X$ which lie in the convex hull of a
set $P = \{p_1,\dots,p_k\} \subset X$.
The intersection of $X$ and such a union of balls can easily be formed by
choosing some point $q$ and taking the set
\begin{align}\label{eq-ball-union}
C_q(P) \equiv \{x\in X : \exists p \in P, r_p[x] \le r_p[q] \}.
\end{align}
In order to reduce false positives, we take the intersection of $C_q$ across
all possible points $q$ as our estimate.
\begin{align}\label{eq-conv-estimate}
\widehat{conv}(P) &= \cap_{q \in X} C_q(P)\\
&= \{ x \in X \hspace{-1ex}: \forall q \hspace{-0.4ex} \in \hspace{-0.4ex} X,
\exists p \hspace{-0.4ex} \in \hspace{-0.4ex} P, r_p[x] \le r_p[q] \}
\end{align}
As the intersection of sets containing $conv(P) \cap X$,
we know $\widehat{conv}(P)$ contains $conv(P) \cap X$.
Theorem~\ref{prop-chull-errors}, proved in Appendex~\ref{app-convex-hull}, says
that
any false positives in our estimate are close to the boundary of $conv(P)$.
See Figure~\ref{fig-frft-axis} for an example of an axis we might select.

\begin{theorem}\label{prop-chull-errors}

Let $\widehat{conv}(A)$ be the estimate of $conv(A)$ for some
$A \subseteq X \subset \mathbb{R}^d$.
If the largest empty ball in $conv(\widehat{conv}(A))$ has radius $\epsilon$,
and the maximum distance between any two points in $A$ is $m$,
then for any $c \in \widehat{conv}(A)$ the distance to the closest point
$c' \in conv(A)$ is less than $\sqrt{ \epsilon(2m + \epsilon) }$.
Further, there is no point $x \in X$ such that $r_a[x] < r_a[c]$ for all
$a \in A$.

\end{theorem}

\begin{figure}[t]
\centering
\includegraphics[width=.35\textwidth]{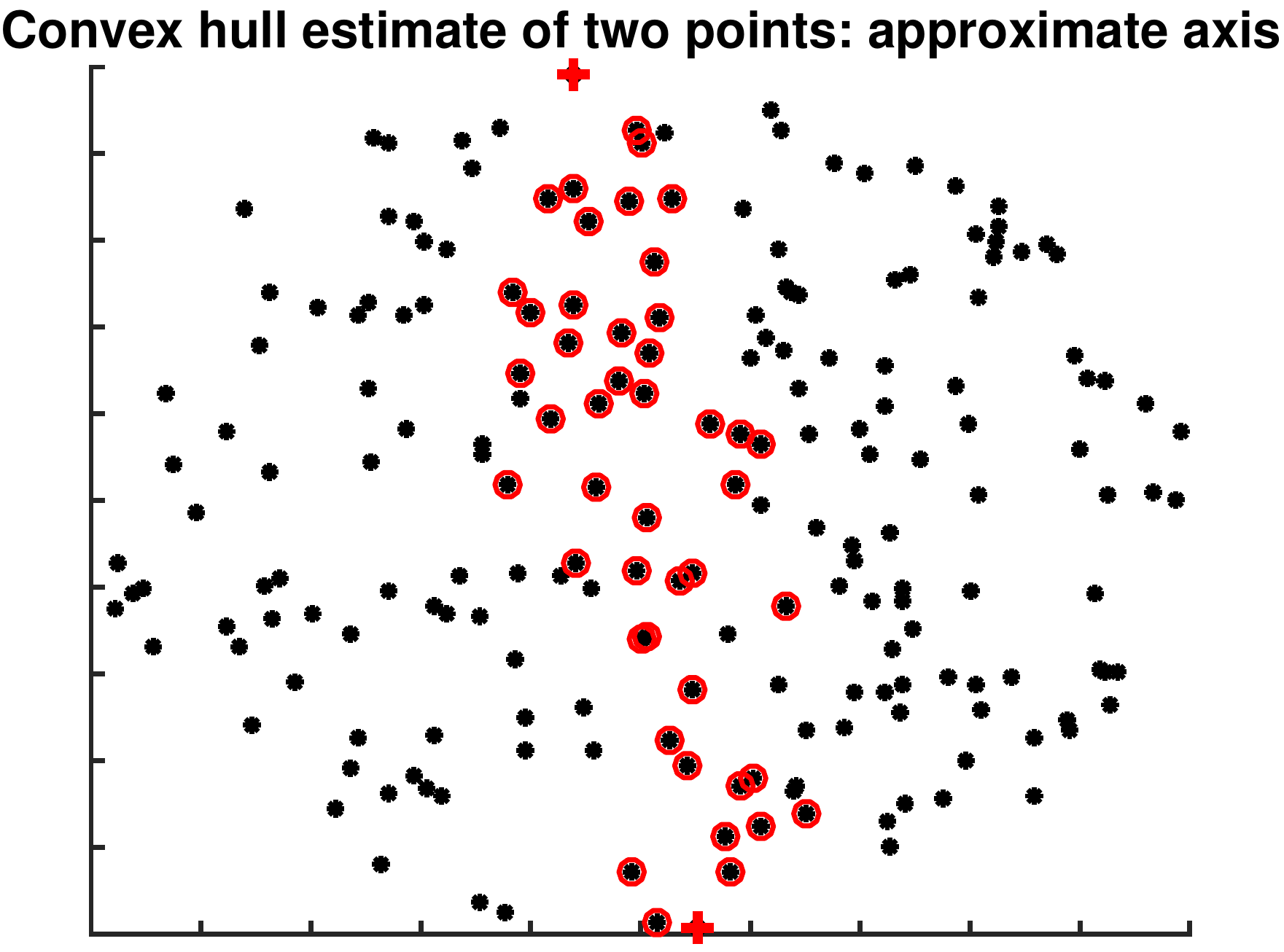}
\captionsetup{justification=raggedright}
\caption{Red points are within $\widehat{conv}$ of the two (blue) endpoints.
Although some of these points are very close together, none is ``above''
any other w.r.t. the endpoints. Lenses are very thin near the line.
}
\label{fig-frft-axis}
\vspace{-1.5em}
\end{figure}

It is easy to show that the points in any axis are order-consistent:
increasing distance order from one axis endpoint is
decreasing distance order from the other endpoint (matching the intuition for
points along a line).
However, the points may be somewhat distant from the line in an arbitrary
direction.

\subsection{Embedding Each Point}

Given an approximate basis $A_1,\dots,A_{\hat{d}}$, the next step is to embed
the points within the basis.
We accomplish this without any additional comparisons.

Ideally, along each axis $A_i$ the points would be evenly-spaced and lie along
the line between the axis endpoints.
We could then embed any $x \in X$ by simply finding
the index of the closest point via binary search, since
the members of $A_i$ would be sorted as a bitonic array: the distance to $x$
would descend to a minimum and then ascend.
The total comparisons cost would be $O(n \hat{d} \log n)$
(within the theoretical bound).

In practice, we never have such perfect axes.
When the points of $X$ are in general position, no
member of $X$ will be found in the convex hull of any subset of $d$ or
fewer points and no member of $A_i$ except the endpoints will lie
on the line.
A binary search will not find the closest point in $A_i$ to
$x$ because the points will be not be exactly sorted by distance to $x$.
Further, the closest point in $A_i$ will often be closest simply because it is
not found on the endpoints-line, not because it is near the projection $x'$ of
$x$ onto the line (Figure~\ref{fig-locate-point}).

\begin{figure}[h]
\centering
\includegraphics[width=.35\textwidth]{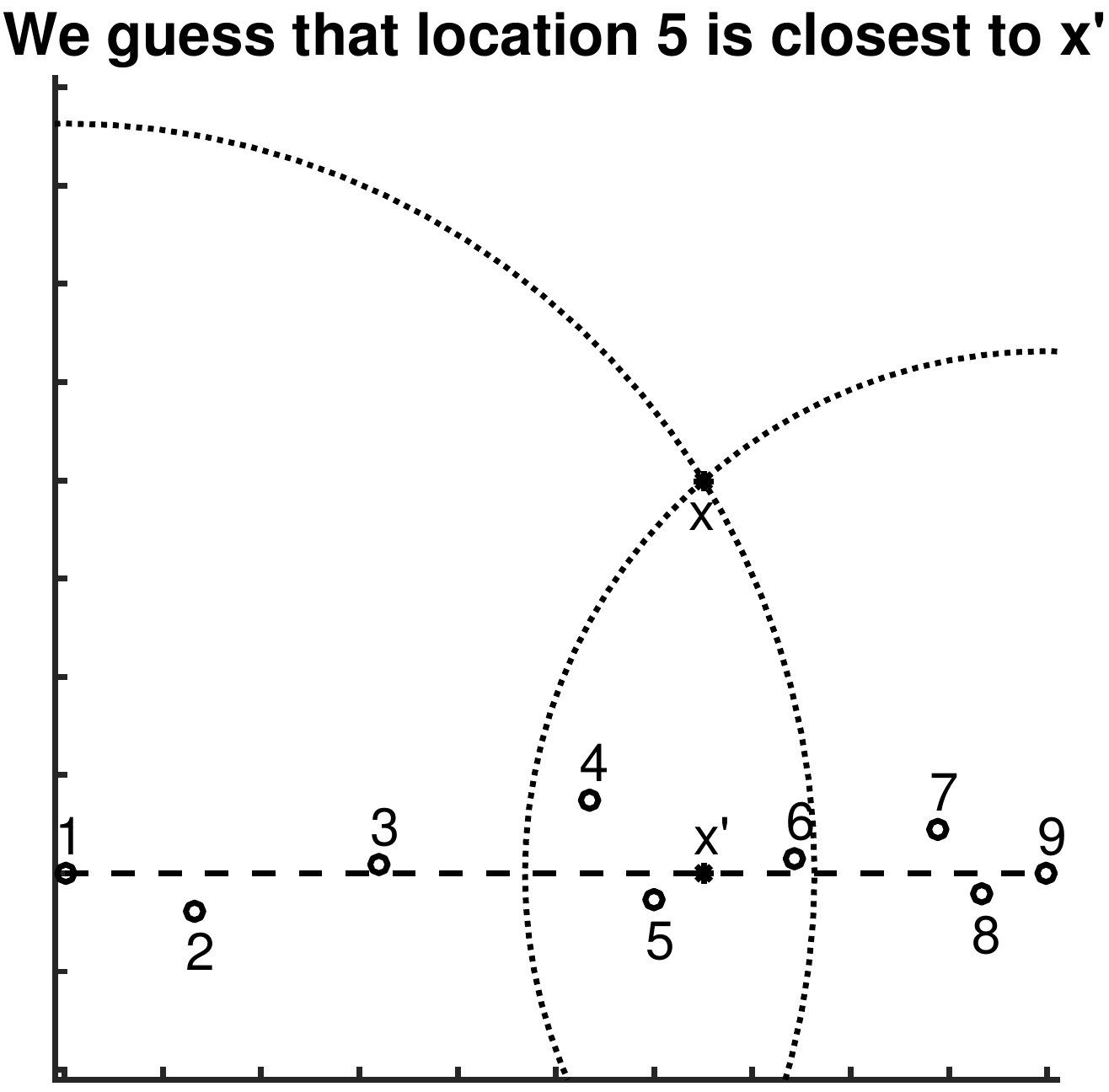}
\captionsetup{justification=raggedright}
\caption{The circled points are members of axis $A_i$.
We choose the median point within the lens beneath $x$ (containing 4, 5, and
6) as our guess at the closest point to its projection, $x'$.
}
\label{fig-locate-point}
\vspace{-1.5em}
\end{figure}

We really want to find the point in $A_i$ which is closest to $x'$,
not closest to $x$.
The projection $x'$ will be in the center of the lens formed
from the axis endpoints with $x$ at its apex.
The lens will always contain some member of $A_i$ (otherwise $x \in A_i$).
We select as the ordinal coordinate of $x$ along
$A_i$ \textbf{the median index} for those axis points inside this lens.

While this may not be the point in $A_i$ closest to $x'$, especially if the
density varies greatly along $A_i$, it costs no
additional comparisons to select this point.
We have found that the empirical performance on our datasets is comparable
to finding the point in the lens closest to $x$ through a linear search.

While even in dense spaces our algorithm might not find orthogonal axes, if the
discovered axes are indeed orthogonal we can guarantee that our embedding
recovers the original metric with a precision depending on the density of $X$
and the true dimensionality.

\begin{theorem}\label{perfect-embed-epsilon}
If any ball of radius $\epsilon$ in $conv(X)$ contains at least 1 and at most
$k$ points, and assuming $\hat{d}=d$ orthogonal axes are found which extend to
the faces of a bounding box for $X$,
using linear search in the lens for points' coordinates, 
then there is a scaling constant $s \in \mathbb{R}$
such that for any two points $x,y \in X$,
\begin{itemize}[nosep,leftmargin=0.3cm]
\item the coordinate $x_i$ on any axis $A_i$ is bounded by its projection
$x'_i$ by
$(s/k) x_i - \epsilon \le x'_i \le s x_i + \epsilon$,
and
\item the scaled distance estimate 
	$\hat{d}:=s \cdot \widehat{dist}(x,y)$ is within
	$2k\epsilon\sqrt{d}$ of the true distance, i.e.,\\
		$dist(x,y) - 2\epsilon\sqrt{d}
		\leq \hat{d}
		\leq k( dist(x,y) + 2\epsilon\sqrt{d} )$.
\end{itemize}
\end{theorem}

This is proven in the appendix.
If we fix $d$ and the diameter of $X$ in each dimension, assuming
orthogonal axes, this theorem implies that when $\epsilon \rightarrow 0$ the
value of $|X|$ approaches infinity and the distances are recovered, down to
scaling, recovering the original metric.

\section{Basis Evaluation Results}
\label{sec-basis-eval}

We show here that our algorithm produces good embeddings on real
datasets, and that an optimization-based embedding of our triples often yields
an even better result.

\paragraph{Datasets.}
We evaluate against several generated and real datasets.
In an attempt to fully exercise our algorithm, we include some datasets which
are not well-suited to it.
We answer all similarity questions based on Euclidean distances between the
original features/positions.
The \texttt{3dgmm}, \texttt{5dgmm}, and \texttt{5dcube} datasets are random
draws of 500 points from Gaussian mixtures and the unit cube.

For \texttt{cities}, we select 500 cities by choosing the most populous city in
each country and then additional cities from most to least populous.
Cities are represented in Euclidean coordinates converted from their
latitudes and longitudes.
We use continents as class labels.
All cities lie on the convex hull, so $\epsilon$ is the diameter of the set.
This leads to a smaller dimensionality estimate and somewhat worse performance.

We use 1,000 records from \texttt{MNIST Digits}, treating raw pixel
values as 784 features.
\texttt{Digits} roughly consists of clusters of points around
class labels, with gaps between them.
We thus have a fairly large $\epsilon$, and wildly varying density
across the dataset.
We also take 1,000 records from \texttt{Spambase} with 57
features.
Performance on \texttt{spambase} is especially strong.
Our basis for this dataset outperforms all the others, although SOE was not
able to find a global optimum from our triples.

Finally, we take 2000 records from \texttt{20newsgroups}, using TF-IDF scores 
for 34,072 terms as features. These records are modified using the
\texttt{scikit-learn} Python package to remove headers, footers, and quotes,
and we only select records having at least 50 nonzero feature values.
This dataset has $d \gg n$ and very sparse feature vectors.
This type of distribution is ill-suited to our algorithm, and performs the
worst.

\begin{table*}[t]
\caption{Embedding Quality
\\
$~^*$ indicates global optimum was not found;
means procedure computationally too expensive
}
\label{tbl-embeddings}
\centering

\begin{tabular}{cc}

\setlength{\tabcolsep}{1pt}
\begin{tabular}[t]{llcccccc}
{\bf Method} & {\bf Dataset} & {$d$} & {$\hat{d}$} & {\# Cmp.}
	& {$\tau$} & {$k$nn} & {rmse}
\\
\hline
Basis & 3dgmm & 3 & 3 & 38K & 0.71 & 0.64 & 0.77 \\
Basis+SOE & 3dgmm & 3 & 3 & 38K & 0.99 & 0.97 & 0.02 \\
Extra+SOE & 3dgmm & 3 & 3 & 61K & {\bf 0.99} & {\bf 0.99} & {\bf 0.01} \\
Rand+SOE & 3dgmm & 3 & 3 & 38K & 0.95 & 0.81 & 0.11 \\
CK & 3dgmm$^*$ & 3 & 3 & 38K & -0.01 & 0.02 & 1.79 \\
\hline
Basis & 5dcube & 5 & 3 & 39K & 0.49 & 0.40 & 0.26 \\
Basis+SOE & 5dcube & 5 & 6 & 39K & 0.88 & 0.73 & 0.05 \\
Extra+SOE & 5dcube & 5 & 6 & 61K & {\bf 0.94} & {\bf 0.92} & {\bf 0.03} \\
Rand+SOE & 5dcube$^*$ & 5 & 6 & 39K & 0.61 & 0.30 & 0.19 \\
CK & 5dcube$^*$ & 5 & 5 & 39K & 0.01 & 0.02 & 0.34 \\
\hline
Basis & 5dgmm & 5 & 3 & 39K & 0.68 & 0.60 & 0.90 \\
Basis+SOE & 5dgmm & 5 & 6 & 39K & 0.94 & 0.66 & 0.14 \\
Extra+SOE & 5dgmm & 5 & 6 & 62K & {\bf 0.98} & {\bf 0.97} & {\bf 0.04} \\
Rand+SOE & 5dgmm$^*$ & 5 & 6 & 39K & 0.01 & 0.02 & 1.77 \\
CK & 5dgmm$^*$ & 5 & 5 & 39K & -0.01 & 0.02 & 1.57 \\
\end{tabular}

&

\setlength{\tabcolsep}{1pt}
\begin{tabular}[t]{llcccccc}
{\bf Method} & {\bf Dataset} & {$d$} & {$\hat{d}$} & {\# Cmp.}
	& {$\tau$} & {$k$nn} & {rmse}
\\
\hline
Basis & 20news & 34K & 3 & 186K & {\bf 0.11} & {\bf 0.06} & 0.53 \\
Basis+SOE & 20news$^*$ & 34K & 6 & 186K & 0.01 & 0.01 & {\bf 0.34} \\
Extra+SOE & 20news$^*$ & 34K & 6 & 310K & -0.01 & 0.01 & 0.34 \\
Rand+SOE & 20news$^*$ & 34K & 3 & 186K & 0.01 & 0.01 & 0.44 \\
CK & 20news & 34K & 16 & --- & --- & --- & --- \\
\hline
Basis & cities & 3 & 2 & 28K & 0.37 & 0.35 & 0.60 \\
Basis+SOE & cities & 3 & 4 & 28K & 0.89 & 0.54 & 0.13 \\
Extra+SOE & cities & 3 & 4 & 50K & {\bf 0.96} & {\bf 0.93} & {\bf 0.05} \\
Rand+SOE & cities$^*$ & 3 & 4 & 28K & 0.01 & 0.02 & 0.75 \\
CK & cities$^*$ & 3 & 3 & 28K & 0.01 & 0.02 & 0.67 \\
\hline
Basis & digits & 784 & 6 & 159K & 0.52 & 0.29 & 3.18 \\
Basis+SOE & digits$^*$ & 784 & 12 & 159K & 0.01 & 0.01 & 2.48 \\
Extra+SOE & digits$^*$ & 784 & 12 & 211K & 0.01 & 0.01 & 2.49 \\
Rand+SOE & digits$^*$ & 784 & 12 & 159K & {\bf 0.73} & {\bf 0.40} & {\bf 2.31}
\\
CK & digits & 784 & 10 & --- & --- & --- & --- \\
\hline
Basis & spam & 57 & 3 & 85K & 0.85 & {\bf 0.78} & 471 \\
Basis+SOE & spam$^*$ & 57 & 6 & 85K & -0.01 & 0.01 & 596 \\
Extra+SOE & spam$^*$ & 57 & 6 & 138K & 0.01 & 0.01 & 596  \\
Rand+SOE & spam & 57 & 3 & 85K & {\bf 0.94} & 0.23 & {\bf 150} \\
CK & spam & 57 & 10 & --- & --- & --- & --- \\
\end{tabular}
\end{tabular}
\vspace{-1.5em}
\end{table*}

\paragraph{Evaluation.}
We evaluate by comparing all rankings or distances between an embedding and
the original dataset.
We report
mean Kendall's $\tau$,
mean $k$NN precision for $k=\lceil \log_2 n \rceil$,
and distance RMSE.

\emph{Distance RMSE} (root mean squared error), 
is based on the fact that in a perfect embedding all pairwise distances would
be scaled by the same constant.
Recall that $d(x,y)$ is the distance between $x$ and $y$ in $X$, and denote
the distance in some embedding $\hat{X}$ by $\hat{d}(x,y)$.
If $\hat{X}$ recovers $X$, then there is some $s \in \mathbb{R}$ such
that for all points $x,y \in X$, $d(x,y) \approx s\hat{d}(x,y)$.
We fit an optimal $s$ and report the RMSE of the residuals,
\begin{align*}\label{eq-distance-rmse}
rmse&(X,\hat{X}) \equiv 
&\min_{s} \left( \frac{1}{n} \sum_{i<j}
(d(x_i,x_j) - s\hat{d}(x_i,x_j))^2 \right)^{1/2}
\end{align*}
Smaller is better and zero is perfect, but the numbers are not comparable
across different datasets.

\paragraph{Experiments.}
Our results are in Table~\ref{tbl-embeddings}.
\texttt{Basis} is the embedding our geometric algorithm produces.
\texttt{Basis+SOE} uses the triples collected by \texttt{Basis} as input to
the Soft Ordinal Embedding (SOE) algorithm \citep{Terada:2014tr}.
\texttt{Extra+SOE} runs use additional comparisons to improve the
embedding as described in Section~\ref{sec-improving}.

SOE does very well when a global optimum is found, but often takes many random
initializations to find one.
We attempt 20 embeddings in $\hat{d}$ dimensions;
if the minimal loss is above $10^{-3}$, we try again in $2\hat{d}$ dimensions
and report the best of 20 embeddings in the higher dimensionality.
Even a small amount of loss from the SOE objective can lead to a
poor embedding, and for several of our datasets it was simply unable to
find a global optimum or even a competitive local optimum (at least, not in 40
attempts).
Note that zero loss is always possible for correct comparisons
in the true dimensionality, but not necessarily in $\hat{d}$ dimensions.
In general, one never knows whether a particular loss threshold can be
achieved, especially given noisy or potentially non-Euclidean comparisons.

\paragraph{Baselines.}
We compare against two baselines.
Notably, \texttt{Basis} is much faster than the baselines, completing in less
than two seconds for each dataset and often much less than one.
In contrast, the SOE runs took more than 24 hours to repeat embeddings with new
random initializations, and the \texttt{CK} runs took several days to generate
triples.

\texttt{Rand+SOE} picks random comparisons in round robin style for each
head until its budget is exhausted, and then embeds them with SOE.
SOE generally struggled to embed these triples, but performance was good
when it worked.
\texttt{CK} runs the CrowdKernel method \citep{Tamuz:2011vma} up to the number
of triples used to build our basis, and evaluates the resulting CrowdKernel
embedding.
We used the authors' CrowdKernel code, but it was not able to handle the
larger datasets.
We report all datasets which completed.

\begin{table}[t]
\caption{Classification Accuracy, 5 folds}
\label{tbl-classification}
\begin{center}
\begin{tabular}{lrrrrrr}
	& \multicolumn{3}{c}{Original}
	& \multicolumn{3}{c}{Embedding}
\\
{\bf Dataset}
& $d$ & Train & Test & $\hat{d}$ & Train & Test \\
\hline
20news & 34K & 0.94 & 0.54 & 3 & 0.21 & 0.08 \\
cities & 3 & 1 & 0.95 & 2 & 0.99 & 0.90 \\
digits & 784 & 1 & 0.84 & 6 & 0.92 & 0.71 \\
spam & 57 & 0.99 & 0.97 & 3 & 0.85 & 0.74 \\
\end{tabular}
\end{center}
\vspace{-1.5em}
\end{table}

As a simple test of downstream utility,
we trained Gradient Boosting classifiers on our basis embeddings and
compared the classification accuracy to classifiers trained on the original
feature space.
The results can be seen in Table~\ref{tbl-classification}.
Performance is good even with $\hat{d}\ll d$, with the exception of the
ill-suited \texttt{20newsgroups} (which is at least better than random).

\section{Improving the Ordinal Embedding}
\label{sec-improving}

Once we have obtained an embedding of reasonably high quality, it is not
difficult to adaptively select new triples to drive the quality upward.
First, we pause to consider the information that can already be inferred
from the triples gathered thus far.
Any geometric properties implied by the triples must be true of any embedding
that satisfies them, so this helps us reason about what we
have already ``told'' the optimizer.

We have sorted all points from the endpoints of each axis, and selected
endpoints that are near the boundaries of the set.
This already carries a lot of information to an embedding algorithm.
The estimates $\widehat{conv}$ are fixed for any subset of axis endpoints,
establishing a layer of points near their hulls.
We know the set of points which are only ``above'' points in $\widehat{conv}$,
establishing a second layer,
and we similarly know the contents of every layer up to the set boundary.

We lack information about dimensions whose extents were too small or sparse for
us to discover.
We don't know exact distances, so we
can't immediately identify the $k$-nearest neighbors of each point.
With the results of \citet{Terada:2014tr} and \citet{Hashimoto:2015wj} in mind,
showing good performance embedding based on the $k$NN with $k = \log n$,
we sort the $2k$ nearest points to each point within our embedding,
costing $\Theta(n k \log k)$ additional comparisons.
For our experiments, we use $k = \log_2 n$ for a total cost of
$\Theta(n \log n \log \log n)$ additional triples.
Note that if you wish to simply identify the $k$NN rather than sort them, a
selection algorithm can instead be used for $\Theta(n \log n)$ new
comparisons in total, staying within the lower bound.
%
%
The result is in Table~\ref{tbl-embeddings} as the \texttt{Extra+SOE} line.
In all cases, the embedding quality improved.

\section{Related Work}
\label{sec-related}

Ordinal Embedding, a.k.a. non-metric embedding or non-metric
multidimensional scaling, has been studied for over sixty years.
Optimal comparison selection, however, is less studied.
When all answers are known in advance (i.e. from features),
practitioners either use them all, select a random subset,
identify the kNN, or sort a set of ``landmark'' objects.
\citet{Jamieson:2011bt} suggest using embeddings to determine whether a
question can be decided from prior answers.
The only adaptive algorithm we have found which works in practice is
the CrowdKernel algorithm by \citet{Tamuz:2011vma}.
Given an intermediate embedding based on prior answers,
it greedily selects new questions for each object to maximize the
expected information gain for the embedding objective.
This method outperforms random selection (apparently by selecting
tails closer than average to the head), but its embeddings compare
poorly to those of Soft Ordinal Embedding.

Ordinal embedding has been heavily studied, particularly by the metric
and kernel learning communities.
Early approaches employed semidefinite programming
\citep{Weinberger06distancemetric, Xing03distancemetric}
and/or required eigenvalue decompositions.
Later approaches focused on minimizing Bregman divergences
\citep{Davis:2007:IML:1273496.1273523, Kulis:2009:LKL:1577069.1577082,
Jain:2012:MKL:2503308.2188402}, which is guaranteed to find a positive
semidefinite (PSD) kernel, or on ignoring semidefiniteness until convergence
and projecting the output matrix to the nearest PSD matrix
\citep{Chechik:2010:LSO:1756006.1756042}.
Ordinal embedding without features has also been studied by
\citet{Agarwal:2007wq,Agarwal:2010hc}, who provide a flexible and
modular algorithm with proven convergence guarantees.
\citet{McFee:2011vt} considers how to learn a similarity function which is as
consistent as possible with multiple feature sets as well as ordinal
constraints.
Local (and Soft) Ordinal Embedding \citep{Terada:2014tr} recovers the metric
with guarantees on accurate density recovery.
\citet{Hashimoto:2015wj} prove metric recovery over certain directed graphs, of
including kNN adjacency graphs.

Our algorithm relies on a sort routine, so when an unreliable oracle is used it
is natural to consider the deep literature on crowdsourcing sort algorithms
\citep{Marcus:2011wf,Niu:2015jc}
and on noise-tolerant sorting
\citep{Ajtai:2009ht,Braverman:2008wz,Hadjicostas:2011dv}.

\section{Conclusion and Future Work}
\label{sec-conclusion}

We have presented a Computational Geometric approach to Ordinal Embedding which
offers new theoretical insights into the problem.
In particular, we have contributed approximate algorithms for
dimensionality estimation, tests of convex hull membership and affine
independence, and perpendicular line discovery.
We have combined these methods to find an approximate basis within which points
can easily be positioned.
When run on a sufficiently dense set of relatively low dimensionality, we can
reliably and efficiently produce a medium-to-high quality embedding.
When an optimizer finds a global optimum for our triples, the user
obtains a high quality embedding.

While we have not ``solved'' the embedding or triple selection problems and do
not suggest replacing optimization approaches entirely, our approach provides
new insights into the geometric information contained in a set of triples and
we believe it will lead to faster and more reliable future approaches.

\clearpage
\bibliography{refs}

\begin{thebibliography}{21}
\providecommand{\natexlab}[1]{#1}
\providecommand{\url}[1]{\texttt{#1}}
\expandafter\ifx\csname urlstyle\endcsname\relax
  \providecommand{\doi}[1]{doi: #1}\else
  \providecommand{\doi}{doi: \begingroup \urlstyle{rm}\Url}\fi

\bibitem[Agarwal et~al.(2010)Agarwal, Phillips, and
  Venkatasubramanian]{Agarwal:2010hc}
Arvind Agarwal, Jeff~M. Phillips, and Suresh Venkatasubramanian.
\newblock Universal multi-dimensional scaling.
\newblock \emph{SIGKDD}, 2010.
\newblock \doi{10.1145/1835804.1835948}.

\bibitem[Agarwal et~al.(2007)Agarwal, Wills, Cayton, Lanckriet, Kriegman, and
  Belongie]{Agarwal:2007wq}
Sameer Agarwal, Josh Wills, Lawrence Cayton, Gert Lanckriet, David Kriegman,
  and Serge Belongie.
\newblock Generalized non-metric multidimensional scaling.
\newblock \emph{AISTATS}, 2007.

\bibitem[Ajtai et~al.(2009)Ajtai, Feldman, Hassidim, and Nelson]{Ajtai:2009ht}
Mikl{\'o}s Ajtai, Vitaly Feldman, Avinatan Hassidim, and Jelani Nelson.
\newblock {Sorting and Selection with Imprecise Comparisons}.
\newblock \emph{ICALP}, 2009.

\bibitem[Arias-Castro(2015)]{AriasCastro:2015wk}
Ery Arias-Castro.
\newblock {Some theory for ordinal embedding}.
\newblock \emph{arXiv.org}, January 2015.

\bibitem[Braverman and Mossel(2008)]{Braverman:2008wz}
Mark Braverman and Elchanan Mossel.
\newblock Noisy sorting without resampling.
\newblock \emph{SODA}, 2008.

\bibitem[Chechik et~al.(2010)Chechik, Sharma, Shalit, and
  Bengio]{Chechik:2010:LSO:1756006.1756042}
Gal Chechik, Varun Sharma, Uri Shalit, and Samy Bengio.
\newblock Large scale online learning of image similarity through ranking.
\newblock \emph{JMLR}, 2010.

\bibitem[Davis et~al.(2007)Davis, Kulis, Jain, Sra, and
  Dhillon]{Davis:2007:IML:1273496.1273523}
Jason~V. Davis, Brian Kulis, Prateek Jain, Suvrit Sra, and Inderjit~S. Dhillon.
\newblock Information-theoretic metric learning.
\newblock \emph{ICML}, 2007.

\bibitem[Gonzalez(1985)]{Gonzalez:1985kk}
Teofilo~F Gonzalez.
\newblock {Clustering to minimize the maximum intercluster distance}.
\newblock \emph{Theoretical Computer Science}, 38:\penalty0 293--306, 1985.

\bibitem[Hadjicostas and Lakshmanan(2011)]{Hadjicostas:2011dv}
Petros Hadjicostas and K~B Lakshmanan.
\newblock {Recursive merge sort with erroneous comparisons}.
\newblock \emph{Discrete Applied Mathematics}, 159\penalty0 (14):\penalty0
  1398--1417, August 2011.

\bibitem[Hashimoto et~al.(2015)Hashimoto, Sun, and Jaakkola]{Hashimoto:2015wj}
T~B Hashimoto, Y~Sun, and T~S Jaakkola.
\newblock {Metric recovery from directed unweighted graphs.}
\newblock \emph{AISTATS}, 2015.

\bibitem[Jain et~al.(2012)Jain, Kulis, Davis, and
  Dhillon]{Jain:2012:MKL:2503308.2188402}
Prateek Jain, Brian Kulis, Jason~V. Davis, and Inderjit~S. Dhillon.
\newblock Metric and kernel learning using a linear transformation.
\newblock \emph{JMLR}, 2012.

\bibitem[Jamieson and Nowak(2011)]{Jamieson:2011bt}
Kevin~G Jamieson and Robert~D Nowak.
\newblock Low-dimensional embedding using adaptively selected ordinal data.
\newblock \emph{49th Annual Allerton Conference on Communication, Control, and
  Computing}, 2011.
\newblock \doi{10.1109/Allerton.2011.6120287}.

\bibitem[Kleindessner and von Luxburg(2014)]{Kleindessner:2014wu}
M~Kleindessner and U~von Luxburg.
\newblock {Uniqueness of Ordinal Embedding.}
\newblock \emph{COLT}, 2014.

\bibitem[Kulis et~al.(2009)Kulis, Sustik, and
  Dhillon]{Kulis:2009:LKL:1577069.1577082}
Brian Kulis, M\'{a}ty\'{a}s~A. Sustik, and Inderjit~S. Dhillon.
\newblock Low-rank kernel learning with bregman matrix divergences.
\newblock \emph{JMLR}, 2009.

\bibitem[Marcus et~al.(2011)Marcus, Wu, Karger, Madden, and
  Miller]{Marcus:2011wf}
Adam Marcus, Eugene Wu, David Karger, Samuel Madden, and Robert Miller.
\newblock Human-powered sorts and joins.
\newblock \emph{Proc. VLDB Endow.}, 2011.

\bibitem[McFee and Lanckriet(2011)]{McFee:2011vt}
Brian McFee and Gert Lanckriet.
\newblock {Learning Multi-modal Similarity}.
\newblock \emph{JMLR}, 12, February 2011.

\bibitem[Niu et~al.(2015)Niu, Lan, Guo, Cheng, Yu, and Long]{Niu:2015jc}
Shuzi Niu, Yanyan Lan, Jiafeng Guo, Xueqi Cheng, Lei Yu, and Guoping Long.
\newblock {Listwise Approach for Rank Aggregation in Crowdsourcing}.
\newblock \emph{WSDM}, 2015.

\bibitem[Tamuz et~al.(2011)Tamuz, Liu, Belongie, Shamir, and
  Kalai]{Tamuz:2011vma}
Omer Tamuz, Ce~Liu, Serge Belongie, Ohad Shamir, and Adam~Tauman Kalai.
\newblock {Adaptively Learning the Crowd Kernel}.
\newblock \emph{arXiv.org}, May 2011.

\bibitem[Terada and von Luxburg(2014)]{Terada:2014tr}
Yoshikazu Terada and Ulrike von Luxburg.
\newblock {Local ordinal embedding}.
\newblock \emph{ICML}, 2014.

\bibitem[Weinberger et~al.(2006)Weinberger, Blitzer, and
  Saul]{Weinberger06distancemetric}
Kilian~Q. Weinberger, John Blitzer, and Lawrence~K. Saul.
\newblock Distance metric learning for large margin nearest neighbor
  classification.
\newblock \emph{NIPS}, 2006.

\bibitem[Xing et~al.(2003)Xing, Ng, Jordan, and Russell]{Xing03distancemetric}
Eric~P. Xing, Andrew~Y. Ng, Michael~I. Jordan, and Stuart Russell.
\newblock Distance metric learning, with application to clustering with
  side-information.
\newblock \emph{NIPS}, 2003.

\end{thebibliography}
\bibliographystyle{plainnat}

\clearpage
\appendix
\appendix
\section{Proofs}

In the following proofs, we use $B(c,r)$ to denote a closed ball in Euclidean
space with center $c$ and radius $r$,
and $conv(x_1,\dots,x_k)$ to denote the convex hull (the set of all convex
combinations) of points $x_1,\dots,x_k$.
We also use $d(x,y)$ to denote the Euclidean distance between points $x$
and $y$.

\subsection{Convex Rankings}
\label{app-convex-rankings}

The proofs in this section show that our approximate convex hull contains all
points from the convex hull, and that any extra points are not too far from the
hull's boundary.
We essentially show that for any point $p$ outside the convex hull of a set $V$
of points, there is some point $q$ inside the convex hull which is closer to
all the members of $V$ than is $q$.

We first prove a useful lemma.
%
%
\begin{lemma}\label{prop-convex-hull-contains-closer}
Let $V = \{v_1,\dots,v_k\}$ be an arbitrary set of points in Euclidean space,
and $p$ an arbitrary point not found inside $conv(V)$.
There is a point $q \in conv(V)$ such that for all $v \in V, d(v,q) < d(v,p)$.
\end{lemma}
\begin{proof}
	Let $C := conv(V)$.
	Because $C$ is convex, there is some unique point $q \in C$ which is
	closer to $p$ than any other point in $C$, so
	for any arbitrary vertex $v \in V$ we have $d(p,v) \ge d(p, q)$.
	Since $C$ is convex and $q$ is the closest point in $C$ to $p$,
	there is a hyperplane passing through $q$ perpendicular to line $pq$
	which separates $C$ from $p$.
	Thus, we must have either $q = v$ or $\angle p q v \ge 90\degree$.
	So edge $vp$ is the longest in $\triangle v p q$ and
	$q$ is closer to $v$ than is $p$.
\end{proof}

Next, we show that any union of balls which all have at least one point in
common will cover the convex hull of the ball centers.
For example, for an arbitrary subset $V \subseteq X \subset \mathbb{R}^d$,
the set of all points from $X$ which are ranked no farther from the members of
$V$ than some common point $p$ will contain all the points in
$conv(V) \cap X$.

%
%
We next prove Theorem~\ref{prop-ball-covering}.

\begin{theorem}\label{prop-ball-covering}

Let $\mathcal{B} = \{ B(v_1,r_1), \dots, B(v_k, r_k) \}$ be a set of closed
balls in $\mathbb{R}^d$ with centers $v_1,\dots,v_k$ and radii
$r_1,\dots,r_k$, respectively.
If all the balls in $\mathcal{B}$ have at least one point in common, then
$conv(\{v_1,\dots,v_k\})$ is a subset of their union.

\end{theorem}

\begin{proof}
	Let $B := \cup_i B(v_i,r_i)$, let $C := conv(v_1,\dots,v_k)$,
	and let $q$ be an arbitrary point in $C$.
	We will prove that $q$ is in at least one of the component balls in $B$.
	By Carath\'{e}odory's Theorem, there is some subset of at most $d+1$ centers
	$v_1,\dots,v_{d+1}$ (relabeled without loss of generality) such that
	$q \in conv(v_1,\dots,v_{d+1})$.
	Let $C' = conv(v_1,\dots,v_{d+1}) \subseteq C$.

	By Lemma~\ref{prop-convex-hull-contains-closer}, there is some point $p$
	contained in each of the $d+1$ balls which is in $C'$.
	We can partition $C'$ into $d+1$ closed convex subsets by replacing each of
	its vertices $v_i$ in turn with $p$.
	These subsets are $d$-simplexes with $d$ ball centers and $p$ as their $d+1$
	vertices.
	Observe that $q$ must fall into one of these subsets (or more, if it falls on
	a boundary).
	Let $P$ be one such $d$-simplex containing $q$, and let $F$ be the face of the
	simplex formed by the $d$ ball centers.
	Since $q$ lies in $P$, another simplex $Q$ can be formed using the same face
	$F$ but with apex $q$ instead of $p$.
	Call the heights of simplexes $P$ and $Q$ the distances from points $p$ and
	$q$ to their respective closest points in $F$, and observe that the height of
	$Q$ is no greater than the height of $P$.
	Since $Q \subseteq P$, they share the same face $F$,
	and the height of $Q$ is less than or equal to the height of $P$,
	there must be some vertex $v_i$ of $F$ such that $d(v_i, q) \le d(v_i, p)$.
	Thus, any $q \in C$ is contained in $B$.
\end{proof}

The following corollary provides a necessary condition for points in the convex
hull of a set which we use for convex hull estimation.
\begin{corollary}[convex hull rankings]\label{prop-convex-hull-rankings}
	Let $V = \{v_1, \dots, v_k\}$ be an arbitrary set of points in Euclidean
	space, and let points $p$ and $q$ be arbitrary points in $conv(V)$.
	Then there is some $v \in V$ such that $d(v,p) \le d(v,q)$.
\end{corollary}
\begin{proof}
	By Theorem~\ref{prop-ball-covering}, the union of balls centered on the
	members of $V$ whose radii extend to $q$ contain $conv(V)$. Since $p$ is
	inside $conv(V)$, it must be within at least one of these balls.
\end{proof}

\subsection{Convex Hull Estimation}
\label{app-convex-hull}

Given an arbitrary set of points $X \subset \mathbb{R}^d$
and any subset $V \subset X$,
we can use Theorem~\ref{prop-ball-covering} to identify members of $X$ which
are close to $conv(V)$ in the sense that they are either in $conv(V)$ or
close to the boundary of $conv(V)$.

For any point $x \in X$, define the set
$C(x) := \cup_{v \in V} \{ y \in X : d(v,y) \le d(v,x) \}$
as an estimate of the convex hull of $V$.
By Theorem~\ref{prop-ball-covering}, $C(x) \subseteq conv(V)$.
However, any individual estimate $C(x)$ will tend to contain many false
positives.
We can form a better estimate $\hat{C} := \cap_{x \in X} C(x)$.
The following theorem shows that the false positives of this estimate contain
only points which are close to the boundary of $conv(V)$.

%
%
We now prove Theorem~\ref{prop-chull-errors}.

\begin{proof}
By construction, for any points $c \in \hat{C}$ and $x \in X$
there is a vertex $v \in V$ such that $d(v,c) \le d(v,x)$.
However, by Lemma~\ref{prop-convex-hull-contains-closer} we know that there are
points in $conv(V)$ which are closer to any $v \in V$ than any member of
$\hat{C}$ which is not in $conv(V)$.
Any point in $\hat{C}$ which is not in $conv(V)$ is there because none of
these points is contained in $X$.

Let $p \in \hat{C}$ be an arbitrary false positive, not contained in $conv(V)$,
and let $q$ be the closest point in $conv(V)$ to $p$.
Note that since $q \notin X$, we know that $q$ is not a member of $V$.
This means that we tend not to make mistakes ``close to the corners'' of
$conv(V)$.

For any vertex $v \in V$ let $r_v := d(v,p)$ be its distance to $p$.
Define the set of points closer to all members of $V$ than $p$ as
\begin{align}
E_p \equiv \cap_{v \in V} B(v,r_v).
\end{align}
We know $E_p$ contains no members of $X$ because $p$ is a false positive.
Since $E_p$ is an intersection of balls, 
when $d(p,q)$ is greater all the radii $r_v$ are also greater and
the size of $E_p$ is greater in all dimensions.

If $d(p,q)$ was sufficiently large, $E_p$ would contain a ball of radius
$\epsilon$ and would thus contain a member of $X$, forming a contradiction.
The remainder of the proof establishes an upper bound on $d(p,q)$ under the
assumption that an $\epsilon$-ball centered at $q$ is not contained in $E_p$.
Refer to Figure~\ref{fig-chull-errors} for a diagram of the following argument.

\begin{figure}[htb]
\centering	
\begin{tikzpicture}
	\coordinate (v) at (-4, -4);
	\coordinate (p)  at (0, 2);
	\coordinate (q) at (0, 0);
	\coordinate (e) at (-4, 0);
	\coordinate (f) at (1.09, 1.09);
	
	\filldraw
		(v) circle [radius=2pt] node [anchor=east] {$v$}
		(p) circle [radius=2pt] node [anchor=east] {$p$}
		(q) circle [radius=2pt] node [anchor=north west] {$q$}
		(f) circle [radius=2pt] node [anchor=west] {$f$}
		;

	\draw [style=dashed]
		(p) -- node [anchor=west] {$h$} (q)
		(q) -- (e)
		(q) -- node [anchor=north west] {$<\epsilon$} (f)
		(v) -- (e)
		(v) -- node [anchor=south east] {$r_v$} (p)
		(v) -- node [anchor=north west] {$w_v$} (q)
		;
	\draw [right angle symbol={e}{v}{q}];
	\draw [right angle symbol={q}{e}{p}];
	\draw [style=dashed]
    	let \p1=($(p)-(v)$),
		    \n1={atan(\y1/\x1)},
		    \p2=($(f)-(v)$),
		    \n2={atan(\y2/\x2)},
		    \n3={veclen(\x1,\y1)}
		in
	    (p) arc (\n1:\n2:\n3);

\end{tikzpicture}
\caption{Our construction for Theorem~\ref{prop-chull-errors}.
We want to maximize $h$ such that $d(q,f) < \epsilon$.
$v\in V$ is a convex hull vertex, and $p \in \hat{C}$ a false
positive.
$q$ is the closest point in $conv(V)$ to $p$, and point $f$
is collinear with $vq$.
We have $h=d(p,q), r_v = d(v,p) = d(v,f)$, and $z=d(v,q)$.
By assumption, $d(q,f) < \epsilon$.
Note that sometimes $v$ and $f$ are on the same face of $conv(V)$ as $q$,
implying that $w=z$.
}
\label{fig-chull-errors}
\end{figure}
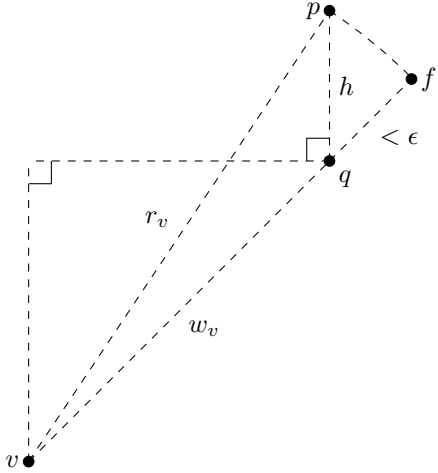

Let $v \in V$ be an arbitrary vertex, not necessarily
on the same face of $conv(V)$ as $q$.
We know that $r_v$ is not large enough that $E_p$ can contain an
$\epsilon$-ball centered at $q$,
so a line segment passing from $v$ to $q$ with length $r_v$ does not extend
by $\epsilon$ or more past $q$.

Let $f$ be the point along the line from $v$ to $q$ at distance $r_v$ from $v$.
We seek the minimum height $h := d(p,q)$ such that $d(q,f) < \epsilon$,
because at this distance an $\epsilon$-ball would (perhaps) fit in $E_p$.

We also define $w_v$ as the distance from $v$ to $q$.

All the points are coplanar, so we proceed using the Pythagorean theorem.
$q$ is on the boundary of $conv(V)$ and $p$ is outside of $conv(V)$, so we
know that $\angle vqp \ge 90\degree$, where we have equality when $v$ is on
the face of $conv(V)$ containing $q$.
In $\triangle vqp$, since $\angle vqp \ge 90\degree$ and $v \ne q$ we have
the following.
\begin{align}
h^2 &\le r_v^2 - w_v^2 \\
&< (w_v + \epsilon)^2 - w_v^2 \\
&= \epsilon^2 + 2w_v\epsilon \\
\implies h &< \sqrt{\epsilon(\epsilon + 2w_v)}
\end{align}


In order to find $h$ such that an $\epsilon$-ball fits in $E_p$, we need
to choose the largest bound on $h$ for any vertex.
For a fixed $X$ we have a fixed $\epsilon$, and the bound scales with
$\sqrt{w_v}$.
It is easy to show that the largest $w_v \le m$, where $m = diam(V)$ is
the maximum pairwise distance between the members of $V$.
This leads to the final bound on $h$.
\end{proof}

\subsection{Basis Quality with high density}

\paragraph {Proof of Theorem~\ref{perfect-embed-epsilon}.}
Let $x,y \in X$ be the embedded (ordinal) positions of arbitrary points in $X$,
let $x_i,y_i$ denote their ordinal coordinates along axis $A_i$, and
let $x'_i,y'_i$ denote the true positions of their projections along the true
(Euclidean) axis between the endpoints of $A_i$.

We first consider the scaling constant $s$, in order to prepare to map between
positions and distances in our ordinal space and in the underlying Euclidean
space.
If the axis points' projections onto the true axes are evenly-spaced, 
then any $\epsilon$-ball centered on the axis will contain exactly $k$ points
and $s = 2 \epsilon/k$ maps from ordinal coordinates to Euclidean coordinates,
i.e., for any $x$ on axis $A_i$,
\begin{align}
	x'_i = s x_i = \frac{2 \epsilon}{k}.
\end{align}

When the axis points' projections are not evenly-spaced, we will have between
one and $k$ points in each $\epsilon$-ball centered on the axis, so we have
that
\begin{align}
	\frac{2\epsilon}{k} x_i \le x'_i \le 2\epsilon x_i.
\end{align}
We will choose $s = 2 \epsilon$ and have that $(s x_i)/k \le x'_i \le s x_i$
for any $x_i$ on axis $A_i$.

For the first claim,
assume that for some $x \in X$ the true position of the embedding coordinate
$x_i$ is more than $\epsilon$ away from the projection $x'_i$.
Since $x_i$ is the closest point inside the lens with apex $x$ formed with
centers in axis endpoints, it means a ball of radius $\epsilon$ fits around the
projection $x'_i$.
Such a ball must include a point $p \in X$, but then $p$ being inside this
$\epsilon$-ball must be in the axis set (as it is between other axis points,
not above them) --- contradiction.
So we have that the true distance between the axis point with index $x_i$
and $x'_i$ is at most $\epsilon$,
and its scaled coordinate $s x_i$ is bounded by
\begin{align}
x'_i \le s x_i + \epsilon, \\
x'_i \ge \frac{s}{k}x_i - \epsilon.
\end{align}

For the second claim, using the bounds from the first claim and assuming
orthogonal axes, we have
\begin{align*}
& (s/k)^2 \cdot \widehat{dist}^2(x,y) \\
= & \sum_{i=1}^d{\left( (s/k)|x_i - y_i| \right)^2}
\\
\leq &\sum_{i=1}^d{(|x'_i - y'_i| + 2\epsilon )^2}
\\
= &\sum_{i=1}^d{|x'_i - y'_i|^2}
	+ 4\epsilon\sum_{i=1}^d{|x'_i - y'_i|}
	+ 4\epsilon^2d
\\
\leq &dist^2(x,y)
	+ 4\epsilon \sqrt{d\sum_{i=1}^d{|x'_i - y'_i|^2}}
	+ 4\epsilon^2d
	&& (*)
\\
= &dist^2(x,y) + 2 \left(dist(x,y) 2\epsilon\sqrt{d}\right)  + 4\epsilon^2d
\\
= &\left(dist(x,y) + 2\epsilon\sqrt{d}\right)^2
\\
\Rightarrow &s \cdot \widehat{dist}(x,y)
	\leq k( dist(x,y) + 2\epsilon\sqrt{d} )
\end{align*}
where $(*)$ follows because the arithmetic mean is less than the square mean.
To summarize, the distances between two points can be considered functions of
their distances along their projections onto the axes, which are correct to
within the specified tolerance.

By the same argument starting with $(1/s) dist(x,y)$ we can show
$dist(x,y) \leq s \cdot \widehat{dist}(x,y) + 2\epsilon\sqrt{d}$ which
concludes the proof.

\subsection{Approximate Basis Quality}
\label{app-basis}

In this section, we provide proofs related to the quality of the approximate
basis found by our algorithm, and its dimensionality estimate.

We begin with the dimensionality estimate, as this drives the upper bound for
our algorithm.

\begin{theorem}[dimensionality estimate]\label{prop-dim-converges}
	Let $\mathcal{X} = \{x_1,x_2,\dots\}$ be an infinite sequence of i.i.d. draws
	from some smoothly-continuous distribution over a simply connected
	compact subset
	$V \subset \mathbb{R}^d$.
	Also let $X_n = \{x_1,\dots,x_n\}$ be the first $n$ draws in $\mathcal{X}$,
	and let $\hat{d}_n$ be the number of axes chosen by
	Algorithm~\ref{fig-alg-basis} when the oracle answers consistently with
	distances between the points in $X_n$.
	Then $\hat{d}_n \le d$ for all $n$, and as
	$n \rightarrow \infty, \hat{d} \rightarrow d$.
\end{theorem}
\begin{proof}
	Since $V$ is bounded and simply connected and is fully supported by the
	distribution, as
	$n \rightarrow \infty$ the radius $\epsilon$ of the largest empty ball in
	$conv(V)$ converges to zero.
	By Theorem~\ref{prop-chull-errors}, this causes our convex hull estimates
	$\widehat{conv}(P) \rightarrow conv(P)$ for any subset $P \subset X_n$.
	For the remainder of the proof, let $P$ and $Q$ be the sets of axis endpoints
	selected by the algorithm in a given iteration of axis selection.

	In some iteration of the algorithm, let $A$ be the set of axis endpoints
	already chosen, and let $p$ be the next axis endpoint selected.
	We choose $p$ as the point ``above'' at least one point and above more points
	than any other candidate.
	It follows from Theorem~\ref{prop-ball-covering} that since $p$ is above some
	other point $q$, it does not lie in $conv(A)$.

	If $p$ is not affinely independent of the members of $A$, the algorithm
	terminates.
	This is because $\widehat{conv}(Q)$ will consist of the union of two
	simplexes:
	one with all vertices in $\widehat{conv}(P)$, and one with $p$ as one vertex
	and all members but one of $P$ as the other vertices.
	The algorithm rejects the new axis if the union of all possible such vertices 
	equals $\widehat{conv}(Q)$.
	Since at most $d+1$ points can be affinely independent in $\mathbb{R}^d$,
	this implies that $\hat{d}_n \le d$ for any $n$.

	If we have selected fewer than $d$ axes, there will be some point $p$ which
	is affinely independent of $A$.
	For sufficiently large $n$, any such $p$ will be above points in
	$\widehat{conv}(Q)$.
	The union of simplexes which we compare to $\widehat{conv}(Q)$ will consist of
	faces of a higher-dimensional simplex.
	If $\epsilon$ is small enough, there will be at least one point in this
	simplex which is not in the union of convex hull estimates, so the algorithm
	will not terminate if $p$ is selected.

\end{proof}

\end{document}